\documentclass[a4paper]{article}
\usepackage{xspace}
\usepackage{epsfig}
\usepackage{amsmath}
\usepackage{amssymb}
\usepackage{amsthm}
\usepackage{color}
\usepackage{url}
\usepackage{graphicx}
\usepackage{booktabs}
\usepackage{lscape}

\newlength{\figurewidth}
\newlength{\smallfigurewidth}

\newtheorem{theorem}{Theorem}
\newtheorem{lemma}{Lemma}
\newtheorem{corollary}{Corollary}
\newtheorem{Remark}{Remark}
\newtheorem{Definition}{Definition}

\usepackage[draft,mode=multiuser]{fixme}
\fxuseenvlayout{colorsig}
\fxusetargetlayout{color}
\FXRegisterAuthor{if}{aif}{IF}
\FXRegisterAuthor{tt}{att}{TT}
\FXRegisterAuthor{yn}{ayn}{YN}
\FXRegisterAuthor{si}{asi}{SI}
\FXRegisterAuthor{hb}{ahb}{HB}
\FXRegisterAuthor{tk}{atk}{TK}
\fxsetface{env}{}

\begin{document}

\title
{\large
	\textbf{MR-RePair: Grammar Compression based on Maximal~Repeats}
}

\author{%
Isamu~Furuya$^{\ast}$,
Takuya~Takagi$^{\ast}$,
Yuto~Nakashima$^{\dag}$,
Shunsuke~Inenaga$^{\dag}$,\\
Hideo~Bannai$^{\dag}$,
Takuya~Kida$^{\ast}$\\[0.5em]
{\small\begin{minipage}{\linewidth}\begin{center}
		\begin{tabular}{ccc}
      $^{\ast}$ Graduate School of IST, & \hspace*{0.5in} & $^{\dag}$ Department of Informatics, \\
      Hokkaido University, Japan & & Kyushu University, Japan\\
      \url{furuya@ist.hokudai.ac.jp} & & \url{yuto.nakashima@inf.kyushu-u.ac.jp}\\
      \url{tkg@ist.hokudai.ac.jp} & & \url{inenaga@inf.kyushu-u.ac.jp}\\
      \url{kida@ist.hokudai.ac.jp} & & \url{bannai@inf.kyushu-u.ac.jp}
		\end{tabular}
	\end{center}\end{minipage}}
}
\date{}

\newcommand{\leaf}{\mathit{leaf}}

\newcommand{\ord}{\mathcal{O}}
\newcommand{\mr}{MR-RePair\xspace}
\newcommand{\mrn}{MR-RePair}
\newcommand{\nmr}{Na\"{\i}ve-MR-RePair\xspace}
\newcommand{\nmrn}{Na\"{\i}ve-MR-RePair}
\newcommand{\rp}{RePair\xspace}
\newcommand{\rpn}{RePair}
\newcommand{\mrorder}{MR-order\xspace}

\newcommand{\grp}{g_{\mathit{rp}}}
\newcommand{\gmr}{g_{\mathit{mr}}}
\newcommand{\gnmr}{g_{\mathit{nmr}}}
\newcommand{\sgrp}{\hat{g}_{\mathit{rp}}}
\newcommand{\sgmr}{\hat{g}_{\mathit{mr}}}
\newcommand{\sgnmr}{\hat{g}_{\mathit{nmr}}}
\newcommand{\Grp}{G_{\mathit{rp}}}
\newcommand{\Gmr}{G_{\mathit{mr}}}
\newcommand{\Gnmr}{G_{\mathit{nmr}}}
\newcommand{\sGrp}{\hat{G}_{\mathit{rp}}}
\newcommand{\sGmr}{\hat{G}_{\mathit{mr}}}
\newcommand{\sGnmr}{\hat{G}_{\mathit{mr}}}

\newcommand{\gram}[1]{G_\mathit{#1}}
\newcommand{\sgram}[1]{\hat{G}_\mathit{#1}}
\newcommand{\gsize}[1]{g_\mathit{#1}}
\newcommand{\sgsize}[1]{\hat{g}_\mathit{#1}}
\newcommand{\var}[1]{V_\mathit{#1}}
\newcommand{\sym}[1]{\Sigma_\mathit{#1}}
\newcommand{\stt}[1]{S_\mathit{#1}}
\newcommand{\rul}[1]{R_\mathit{#1}}
\newcommand{\svar}[1]{\hat{V}_\mathit{#1}}
\newcommand{\ssym}[1]{\hat{\Sigma}_\mathit{#1}}
\newcommand{\sstt}[1]{\hat{S}_\mathit{#1}}
\newcommand{\srul}[1]{\hat{R}_\mathit{#1}}
\newcommand{\set}[1]{\{#1\}}

\newcommand{\fp}{\mathit{fp}}
\newcommand{\fr}{\mathit{fr}}
\newcommand{\vrp}{v_{\mathit{rp}}}
\newcommand{\vmr}{v_{\mathit{mr}}}
\newcommand{\vnmr}{v_{\mathit{nmr}}}
\newcommand{\derive}[1]{\mathit{val}({#1})}

\newcommand{\sstr}[1]{s_\mathit{#1}}
\newcommand{\cnt}[1]{{\rm cnt}({#1})}

\newcommand{\tfrp}[1]{T^{({#1})}_{\mathit{rp}}}
\newcommand{\tfmr}[1]{T^{({#1})}_{\mathit{mr}}}
\newcommand{\rfrp}[1]{\hat{g}^{({#1})}_{\mathit{rp}}}
\newcommand{\rfmr}[1]{\hat{g}^{({#1})}_{\mathit{mr}}}
\newcommand{\wfrp}[1]{\tau^{({#1})}_{\mathit{rp}}}
\newcommand{\wfmr}[1]{\tau^{({#1})}_{\mathit{mr}}}
\newcommand{\trp}{t_{\mathit{rp}}}
\newcommand{\tmr}{t_{\mathit{mr}}}

\newcommand{\occ}[1]{\#{\rm occ}({#1})}
\newcommand{\vari}[1]{[\![{#1}]\!]}
\newcommand{\non}[1]{\bar{#1}}
\newcommand{\darrow}{\overset{*}{\Rightarrow}}

\maketitle
\thispagestyle{empty}

\begin{abstract}
We analyze the grammar generation algorithm of the \rp compression algorithm,
and show the relation between a grammar generated by \rp and maximal repeats.
We reveal that \rp replaces step by step the most frequent pairs within the corresponding most frequent maximal repeats.
Then, we design a novel variant of \rp, called \mr,
which substitutes the most frequent maximal repeats at once instead of substituting the most frequent pairs consecutively.
We implemented \mr and compared the size of the grammar generated by \mr to that by \rp on several text corpora.
Our experiments show that \mr generates more compact grammars than \rp does, especially for highly repetitive texts.
\end{abstract}

\section{Introduction}

Grammar compression is a method of lossless data compression 
that reduces the size of a given text 
by constructing a small context free grammar that uniquely derives the text.
While the problem of generating the smallest such grammar is NP-hard~\cite{1459058},
several approximation techniques have been proposed.
Among them, \rp~\cite{892708} is known as an off-line method 
that achieves a high compression ratio in practice~\cite{claude2010fast, gonzalez2007compressed, wan2003browsing},
despite its simple scheme.
There have been many studies concerning \rp, such as
extending it to an online algorithm~\cite{7786179},
improving its practical working time or space~\cite{7921912, sekine2014adaptive},
applications to other fields~\cite{claude2010fast, lohrey2013xml, tabei2016scalable},
and analyzing the generated grammar size theoretically~\cite{1459058,navarro2008re,Ochoa2018}.

Recently, maximal repeats have been considered as
a measure for estimating how repetitive a given string is:
Belazzougui et al.~\cite{10.1007/978-3-319-19929-0_3} showed that
the number of extensions of maximal repeats is an upper bound on
the number of runs in the Burrows-Wheeler transform
and the number of factors in the Lempel-Ziv parsing.
Also, several index structures whose size is bounded
by the number of extensions of maximal repeats
have been proposed~\cite{belazzougui2017fast, DBLP:conf/cpm/BelazzouguiC17, 10.1007/978-3-319-67428-5_26}.

In this paper,
we analyze the properties of \rp with regard to its relationship to maximal repeats.
As stated above, several works have studied \rp,
but, to the best of our knowledge,
none of them associate \rp with maximal repeats.
Moreover,
we propose a grammar compression algorithm,
called \mr,
that focuses on the property of maximal repeats.
Ahead of this work,
several off-line grammar compression schemes focusing on (non-maximal) repeats 
have been proposed~\cite{apostolico2000off, inenaga2003linear, 4148751}.
Very recently,
Ga{\'n}czorz and Je{\.z} addressed to heuristically improve the compression ratio of \rp 
with regard to the grammar size~\cite{ganczorz2017improvements}.
However, none of these techniques use the properties of maximal repeats.
We show that, under a specific condition,
there is a theoretical guarantee that
the size of the grammar generated by \mr is smaller than or equal to that 
generated by \rp.
We also confirmed the effectiveness of \mr compared to \rp through computational experiments.

\noindent
{\bf Contributions:} The primary contributions of this study are as follows.
\begin{enumerate}
	\item 
	We analyze \rp and show the relation between a grammar generated by \rp and maximal repeats.
	\item
	We design a novel variant of \rp called \mr,
		which is based on substituting the most frequent maximal repeats.
	\item
	We implemented our \mr algorithm and experimentally confirmed 
		that \mr reduces the size of the generated grammar compared to \rp;
		in particular, the size decreased to about 55\% for a highly repetitive text that we used in our experiment.
\end{enumerate}

The remainder of this paper is organized as follows.
In Section~\ref{sec:prelim}, 
we review the notations of strings and the definitions of maximal repeats, grammar compression, and \rp.
In Section~\ref{sec:repair},
we analyze \rp and show the relation between \rp and maximal repeats.
In Section~\ref{sec:mrrepair},
we define \mr, 
compare it with \rp, 
and describe the implementation of it.
In Section~\ref{sec:exp},
we present experimental results of comparison to \rp.
Finally, we conclude the paper in Section~\ref{sec:concl}.

\section{Preliminaries}\label{sec:prelim}

In this sections, we provide some notations and definitions to be used in the following sections.
In addition, we recall grammar compression and review the \rp.

\subsection{Basic notations and terms}
Let $\Sigma$ be an {\em alphabet}, which is an ordered finite set of symbols.
An element $T = t_1 \cdots t_n$ of $\Sigma^{*}$ is called a {\em string}, 
where $|T| = n$ denotes its length.
We denote the empty string by $\epsilon$ which is the string of length $0$, namely, $|\epsilon| = 0$.
Let $\Sigma^{+} = \Sigma^{*} \backslash \{{\epsilon}\}$.
A string is also called a {\em text}.
Let $T = t_1 \cdots t_n \in \Sigma^{n}$ be any text of length $n$.
If $T = usw$ with $u,s,w\in\Sigma^{*}$, then $s$ is called a {\em substring} of $T$.
Then, for any $1 \le i \le j \le n$,
let $T[i..j] = t_i \cdots t_j$ denote the substring of $T$ that begins and ends at positions $i$ and $j$ in $T$,
and let $T[i] = t_i$ denote the $i$th symbol of $T$.
For a finite set $S$ of texts,
text $T$ is said a {\em superstring} of $S$,
if $T$ contains all texts of $S$ as substrings.
We call the number of occurrences of $s$ in a text as a substring, the {\em frequency} of $s$,
and denote it by $\occ{s}$.
Texts $\Sigma^{*}$ and $\hat{\Sigma}^{*}$ are said to be isomorphic for alphabet $\Sigma$ and $\hat{\Sigma}$,
if there exists an isomorphism from $\Sigma$ to $\hat{\Sigma}$.

\subsection{Maximal repeats}
Let $s$ be a substring of text $T$.
If the frequency of $s$ is greater than $1$,
$s$ is called a {\em repeat}.
A {\em left} (or {\em right}) {\em extension} of $s$ is 
any substring of $T$ with the form $ws$ (or $sw$), where $w \in \Sigma^{*}$.
We say that $s$ is {\em left} (or {\em right}) {\em maximal}
if left (or right) extensions of $s$ occur strictly fewer times in $T$ than $s$, and
call $s$ a {\em maximal repeat} of $T$
if $s$ is left and right maximal.
We call $s$ a {\em maximal repeatg} of $T$
if both left- and right-extensions of $s$ occur strictly fewer times in $T$ than $s$. 
In this thesis, we consider only such strings with length more than $1$
as maximal repeats.
For example, substring \texttt{abra} of $T=$\texttt{abracadabra} is a maximal repeat,
while \texttt{br} is not.

\subsection{Grammar compression}
A {\em context free grammar} (CFG or {\em grammar}, simply) $G$ is defined as
a 4-tuple $G = \{V, \Sigma, S, R\}$,
where $V$ is an ordered finite set of {\em variables},
$\Sigma$ is an ordered finite alphabet,
$R$ is a finite set of binary relations called {\em production rules} (or {\em rules})
between $V$ and $(V \cup \Sigma)^{*}$,
and $S\in V$ is a special variable called {\em start variable}.
A production rule represents the manner in which a variable is substituted and written in 
a form $v \rightarrow w$
with $v \in V$ and $w \in (V \cup \Sigma)^{*}$.
Let $X, Y \in {V\cup\Sigma}^{*}$.
If there are $x_l,x,x_r,y\in(V \cup \Sigma)^{*}$ such that $X=x_lxx_r$, $Y=x_lyx_r$, and $x\rightarrow y\in R$,
we write $X\Rightarrow Y$, and denote the reflexive transitive closure of $\Rightarrow$ by $\darrow$.
Let $\derive{v}$ be the string derived from $v$, i.e., $v \darrow \derive{v}$,
and let $\vari{w}$ denote a variable that derives $w$, i.e. $\vari{w}\darrow w$.
Note that $\vari{w}$ is not necessarily unique.
We define grammar 
$\hat{G} = \{\hat{V}, \hat{\Sigma}, \hat{S}, \hat{R}\}$ 
as a {\em subgrammar} of $G$
if $\hat{V} \subseteq V$, $\hat{\Sigma} \subseteq V\cup\Sigma$, and $\hat{R} \subseteq R$.

Given a text $T$, 
{\em grammar compression} is a method of lossless text data compression
that constructs a restricted CFG,
which uniquely derives the given text $T$.
For $G$ to be deterministic, 
a production rule for each variable $v\in V$ must be unique.
In what follows, we assume that every grammar is deterministic and
each production rule is 
$v_i \rightarrow \mathit {expr}_i$,
where $\mathit{expr_i}$ is an expression 
either $\mathit{expr}_i = a$ ($a \in \Sigma$) 
or $\mathit{expr}_i = v_{j_1}v_{j_2}\cdots v_{j_n}$ (
$i > j_k$ for all $1\leq k \leq j_n$).

We estimate the effectiveness of compression by the size of generated grammar,
which is counted by the total length of the right-hand-side of all production rules of the generated grammar.

\subsection{\rp}
\rp is a grammar compression algorithm
proposed by Larsson and Moffat~\cite{892708}.
For input text $T$, let $G=\set{V,\Sigma ,S,R}$ be the grammar generated by \rp.
\rp constructs $G$ by the following steps:\\

	\noindent
	{\bf Step~1.} Replace each symbol $a\in \Sigma$ with a new variable $v_a$ and add $v_a\rightarrow a$ to $R$.\\
	{\bf Step~2.} Find the most frequent pair $p$ in $T$.\\
	{\bf Step~3.} Replace every occurrence (or, as many occurrences as possible, when $p$ is a pair consisting of the same symbol) of $p$ with a new variable $v$,
	then add $v \rightarrow p$ to $R$.\\
	{\bf Step~4.} Re-evaluate the frequencies of pairs	for the renewed text generated in {\bf Step~3}.
	If the maximum frequency is 1,
	add $S \rightarrow {\rm (current~text)}$ to $R$, and terminate.
	Otherwise, return to {\bf Step~2}. \\ 

\begin{figure*}[tb]
\centering
\begin{tabular}{c|c|c|c|c|c|c|c|c|c|c|c|}
\cline{2-12}
	& {\tt a} & {\tt b} & {\tt r} & {\tt a} & {\tt c} & {\tt a} & {\tt d} & {\tt a}  & {\tt b} & {\tt r} & {\tt a} \\ \cline{2-12} \noalign{\vspace{2pt}} \hline
	\multicolumn{1}{|c||}{$v_{\alpha} \rightarrow \alpha~(\alpha = {\rm {\tt a,b,r,c,d}})$}
& $v_{\rm {\tt a}}$ & $v_{\rm {\tt b}}$ & $v_{\rm {\tt r}}$ & $v_{\rm {\tt a}}$ & $v_{\rm {\tt c}}$ & $v_{\rm {\tt a}}$ & $v_{\rm {\tt d}}$ & $v_{\rm {\tt a}}$ & $v_{\rm {\tt b}}$ & $v_{\rm {\tt r}}$ & $v_{\rm {\tt a}}$ \\ \hline
\multicolumn{1}{|c||}{$v_1 \rightarrow v_{\rm {\tt a}}v_{\rm {\tt b}}$} 
	& \multicolumn{2}{c|}{$v_1$} & $v_{\rm {\tt r}}$ & $v_{\rm {\tt a}}$ & $v_{\rm {\tt c}}$ & $v_{\rm {\tt a}}$ & $v_{\rm {\tt d}}$ 
	& \multicolumn{2}{c|}{$v_1$} & $v_{\rm {\tt r}}$ & $v_{\rm {\tt a}}$ \\ \hline
\multicolumn{1}{|c||}{$v_2 \rightarrow v_1v_{\rm {\tt r}}$}                         
	& \multicolumn{3}{c|}{$v_2$} & $v_{\rm {\tt a}}$ & $v_{\rm {\tt c}}$ & $v_{\rm {\tt a}}$ & $v_{\rm {\tt d}}$ 
	& \multicolumn{3}{c|}{$v_2$}         & $v_{\rm {\tt a}}$ \\ \hline
\multicolumn{1}{|c||}{$v_3 \rightarrow v_2v_{\rm {\tt a}}$} & \multicolumn{4}{c|}{$v_3$} & $v_{\rm {\tt c}}$ & $v_{\rm {\tt a}}$ & $v_{\rm {\tt d}}$ 
	& \multicolumn{4}{c|}{$v_3$}                 \\ \hline
\multicolumn{1}{|c||}{$S \rightarrow v_3v_{\rm {\tt c}}v_{\rm {\tt a}}v_{\rm {\tt d}}v_3$} & \multicolumn{11}{c|}{$S$} \\ \hline
\end{tabular}
\caption{An example of the grammar generation process of \rp for text {\tt abracadabra}.
	The generated grammar is 
	$\set{
		\set{v_{\rm {\tt a}},v_{\rm {\tt b}},v_{\rm {\tt r}},v_{\rm {\tt c}},v_{\rm {\tt d}}, v_1, v_2, v_3, S},
		\set{{\rm {\tt a},{\tt b},{\tt r},{\tt c},{\tt d}}}, S, 
		\set{v_{\rm {\tt a}}\rightarrow {\rm {\tt a}},
			v_{\rm {\tt b}}\rightarrow {\rm {\tt b}},
			v_{\rm {\tt r}}\rightarrow {\rm {\tt r}},
			v_{\rm {\tt c}}\rightarrow {\rm {\tt c}},
			v_{\rm {\tt d}}\rightarrow {\rm {\tt d}},
			v_1 \rightarrow v_{\rm {\tt a}}v_{\rm {\tt b}},
			v_2 \rightarrow v_1v_{\rm {\tt c}},
			v_3 \rightarrow v_2v_{\rm {\tt d}},
			S \rightarrow v_3v_{\rm {\tt c}}v_{\rm {\tt a}}v_{\rm {\tt d}}v_3}
	}$,
	and the grammar size is $16$.
	}%
\label{fig:abrarp}
\end{figure*}

Figure~\ref{fig:abrarp} shows an example of the grammar generation process of \rp.

\begin{lemma}[\cite{892708}]\label{lem:rp}
	\rp works in $\ord (n)$ expected time and 
	$5n + 4k^2 + 4k^{\prime} + \lceil \sqrt{n + 1} \rceil - 1$ words of space,
	where $n$ is the length of the source text,
	$k$ is the cardinality of the source alphabet, and
	$k^{\prime}$ is the cardinality of the final dictionary.
\end{lemma}

\section{Analyzing \rp}\label{sec:repair}

In this section, 
we analyze \rp with regard to its relationship to maximal repeats,
and introduce an important concept, called \mrorder.

\subsection{\rp and maximal repeats}

The following theorem shows an essential property of \rp.
That is, \rp recursively replaces the most frequent maximal repeats.

\begin{theorem}\label{theo:rpmr}
    Let $T$ be a given text,
    and assume that every most frequent maximal repeat of $T$ does not appear with overlaps with itself.
	Let $f$ be the frequency of the most frequent pairs of $T$, and
	$t$ be a text obtained after all pairs with frequency $f$ in $T$ are replaced by variables.
	Then, there is a text $s$ such that 
	$s$ is obtained after all maximal repeats with frequency $f$ in $T$ are replaced by variables,
	and $s$ and $t$ are isomorphic to each other.
\end{theorem}

We need two lemmas and a corollary to prove Theorem~\ref{theo:rpmr}.
The following lemma shows a fundamental relation between the most frequent maximal repeats and
the most frequent pairs in a text.

\begin{lemma}\label{lem:pm}
	A pair $p$ of variables is most frequent in text $T$
	if and only if
	$p$ occurs once in exactly one of the most frequent maximal repeats of $T$.
\end{lemma}

\begin{proof}
	($\Rightarrow$)
	Let $r$ be a most frequent maximal repeat that contains $p$ as a substring.
	It is clear that $p$ can only occur once in $r$, since otherwise,
	$\occ{p} > \occ{r}$ would hold, implying the existence of a frequent maximal repeat that is more frequent than $r$,
	contradicting the assumption that $r$ is most frequent.
	Suppose there exists a different most frequent maximal repeat $r'$ that contains $p$ as a substring.
	Similarly, $p$ occurs only once in $r'$.
	Furthermore, since $r$ and $r'$ can be obtained by left and right extensions to $p$,
	$\occ{r} = \occ{r'} = \occ{p}$, and any occurrence of $p$ is contained in an occurrence of both $r$ and $r'$.
	Since $r'$ cannot be a substring of $r$, there exists some string $w$
	that is a superstring of $r$ and $r'$, such that $\occ{w} = \occ{r} = \occ{r'} = \occ{p}$.
	However, this contradicts that $r$ and $r'$ are maximal repeats.

	($\Leftarrow$)
	Let $r$ be the most frequent maximal repeat such that $p$ occurs once in.
	By definition, $\occ{r} = \occ{p}$.
	If $p$ is not the most frequent symbol pair in $T$, there exists some symbol pair $p'$ in $T$
	such that $\occ{p'} > \occ{p} = \occ{r}$. 
	However, this implies that there is a maximal repeat $r'$ with $\occ{r'} = \occ{p'} > \occ{r}$,
	contradicting that $r$ is most frequent.
\end{proof}

The following corollary is directly derived from Lemma~\ref{lem:pm}.

\begin{corollary}\label{coro:pm}
	For a given text,
	the frequency of the most frequent pairs and that of the most frequent maximal repeats are the same.
\end{corollary} 

The following lemma shows an important property of the most frequent maximal repeats.

\begin{lemma}\label{lem:lenol}
	The length of overlap between any two occurrences of most frequent maximal repeats is at most 1.
\end{lemma} 

\begin{proof}
	Let $xw$ and $wy$ be most frequent maximal repeats that
	have an overlapping occurrence $xwy$,
	where $x,y,w\in\Sigma^{+}$.
	If $|w| \geq 2$, then, since $xw$ and $wy$ are most frequent maximal repeats,
	it must be that $\occ{w} = \occ{xw} = \occ{wy}$, i.e., every occurrence of
	$w$ is preceded by $x$ and followed by $y$.
	This implies that $\occ{xwy} = \occ{xw} = \occ{wy}$ as well,
	but contradicts that $xw$ and $wy$ are maximal repeats.	
\end{proof} 

From the above lemmas and a corollary, now we can prove Theorem~\ref{theo:rpmr}.

\vspace{-6pt}
\begin{proof}[Proof of Theorem~\ref{theo:rpmr}]
	By Corollary~\ref{coro:pm}, the frequency of the most frequent maximal repeats in $T$ is $f$.
	Let $p$ be one of the most frequent pairs in $T$.
	By Lemma~\ref{lem:pm}, there is a unique maximal repeat that is most frequent and contains $p$ once.
	We denote such maximal repeat by $r$.
	Assume that there is a substring $zxpyw$ in $T$,
	where $z,w\in \Sigma$, $x,y\in \Sigma^{*}$, and $xpy = r$.
	We denote $r[1]$ and $r[|r|]$ by $\dot{x}$ and $\dot{y}$, respectively.
	There are 2 cases to consider:

	\noindent {\bf (i) $\occ{z\dot{x}} < f$ and $\occ{\dot{y}w} < f$.}
	If $|r| = 2$, the replacement of $p$ directly corresponds to 
	the replacement of the most frequent maximal repeat, since $p = r$.
	If $|r| > 2$, after $p$ is replaced with a variable $v$, $r$ is changed to $xvy$.
	This occurs $f$ times in the renewed text,
	and by Lemma~\ref{lem:pm}, the frequency of every pair occurring in $xvy$ 
	is still $f$.
	Because the maximum frequency of pairs does not increase, $f$ is still the maximum frequency.
	Therefore, we replace all pairs contained in $xvy$ in the following steps,
	and $z\dot{x}$ and $\dot{y}w$ are not replaced.
	This holds for every occurrence of $p$,
	implying that replacing the most frequent pairs while the 
	maximum frequency does not change, 
	corresponds to replacing all pairs contained (old and new)
	in most frequent maximal repeats of the same frequency
	until they are replaced by a single variable.
	Then, we can generate $s$ by replacement of $r$.

	\noindent {\bf (ii) $\occ{z\dot{x}} = f$ or $\occ{\dot{y}w} = f$.}
	We consider the case where $\occ{z\dot{x}} = f$.
	Note that $\occ{zxpy} < f$ by assumption that $xpy$ is a maximal repeat.
	Suppose \rp replaces $z\dot{x}$ by a variable $v$ before $p$ is replaced.
	Note that by Lemma~\ref{lem:pm}, there is a maximal repeat occurring $f$ times and including $z\dot{x}$ once
	(we denote the maximal repeat by $r'$),
	and $r' \neq r$ by assumption.
	By Lemma~\ref{lem:lenol}, the length of overlap of $r$ and $r'$ is at most 1, then only $\dot{x}$ is
	a symbol contained both $r$ and $r'$.
	After that, $xpy=r$ is no longer the most frequent maximal repeat
	because some of its occurrences are changed to $vr[2..|r|]$.
	However, $r[2..|r|]$ still occurs $f$ times in the renewed text.
	Since $\occ{zxpy}<f$ and $\occ{xpy}=f$, $\occ{vr[2]}<f$ and $r[2..|r|]$ is a maximal repeat.
	Then, $r[2..|r|]$ will become a variable in subsequent steps, similarly to {\bf (i)}.
	Here, $r'$ would also become a variable.
	Thus, we can generate $s$ in the way that we replace $r'$ first, then we replace $r[2..|r|]$.
	This holds similarly for $\dot{y}w$ when $\occ{\dot{y}w}=f$,
	and when $\occ{z\dot{x}} = \occ{\dot{y}w}=f$.
\end{proof}
\vspace{-6pt}

\subsection{\mrorder}

From Theorem~\ref{theo:rpmr}, if the most frequent maximal repeat is unique in the current text,
then all the occurrences of it are replaced step by step by \rp.
However, it is a problem if there are two or more most frequent maximal repeats and some of them overlap.
In this case, which maximal repeat is first summarized up depends on
the order in which the most frequent pairs are selected.
Note, however, if there are multiple most frequent pairs, 
which pair is first replaced depends on the implementation of \rp.
We call this order of selecting (summarizing) maximal repeats
{\em maximal repeat selection order} (or {\em \mrorder}, simply).

For instance,
consider a text {\tt abcdeabccde}.
{\tt abc} and {\tt cde} are the most frequent maximal repeats occurring 2 times.
There are 2 \mrorder,
depending on which one is attached priority to the other.
The results after replacement by \rp with the \mrorder are
(i) $xyx{\tt c}x$ with variables $x$ and $y$ such that $x\darrow {\tt abc}$ and $y\darrow {\tt de}$, and
(ii) $zwz{\tt c}w$ with variables $z$ and $w$ such that $z\darrow {\tt ab}$ and $w\darrow {\tt cde}$.
More precisely,
there are 12 possible ways in which \rp can compress the text,
and the generated rule sets are:
\begin{enumerate}
	\item $\set{v_1\rightarrow {\tt ab}, v_2\rightarrow v_1{\tt c}, v_3\rightarrow {\tt de}, 
		S\rightarrow v_2v_3v_2{\tt c}v_3}$,
	\item $\set{v_1\rightarrow {\tt ab}, v_2\rightarrow {\tt de}, v_3\rightarrow v_1{\tt c},
		S\rightarrow v_3v_2v_3{\tt c}v_2}$,
	\item $\set{v_1\rightarrow {\tt bc}, v_2\rightarrow {\tt a}v_1, v_3\rightarrow {\tt de}, 
		S\rightarrow v_2v_3v_2{\tt c}v_3}$,
	\item $\set{v_1\rightarrow {\tt bc}, v_2\rightarrow {\tt de}, v_3\rightarrow {\tt a}v_1, 
		S\rightarrow v_3v_2v_3{\tt c}v_2}$,
	\item $\set{v_1\rightarrow {\tt ed}, v_2\rightarrow {\tt ab}, v_3\rightarrow v_2{\tt c}, 
		S\rightarrow v_3v_1v_3{\tt c}v_1}$,
	\item $\set{v_1\rightarrow {\tt ed}, v_2\rightarrow {\tt bc}, v_3\rightarrow {\tt a}v_2, 
		S\rightarrow v_3v_1v_3{\tt c}v_1}$,
	\item $\set{v_1\rightarrow {\tt ab}, v_2\rightarrow {\tt cd}, v_3\rightarrow v_2{\tt e}, 
		S\rightarrow v_1v_3v_1{\tt c}v_3}$,
	\item $\set{v_1\rightarrow {\tt ab}, v_2\rightarrow {\tt de}, v_3\rightarrow {\tt c}v_2,
		S\rightarrow v_1v_3v_1{\tt c}v_3}$,
	\item $\set{v_1\rightarrow {\tt cd}, v_2\rightarrow {\tt ab}, v_3\rightarrow v_1{\tt e}, 
		S\rightarrow v_2v_3v_2{\tt c}v_3}$,
	\item $\set{v_1\rightarrow {\tt cd}, v_2\rightarrow v_1{\tt e}, v_3\rightarrow {\tt ab}, 
		S\rightarrow v_3v_2v_3{\tt c}v_2}$,
	\item $\set{v_1\rightarrow {\tt ed}, v_2\rightarrow {\tt ab}, v_3\rightarrow {\tt c}v_1, 
		S\rightarrow v_2v_3v_2{\tt c}v_3}$,
	\item $\set{v_1\rightarrow {\tt ed}, v_2\rightarrow {\tt c}v_1, v_3\rightarrow {\tt ab},
		S\rightarrow v_3v_2v_3{\tt c}v_2}$.
\end{enumerate} 
Here, 1 - 6 have the same \mrorder,
because ${\tt abc}$ is prior to ${\tt cde}$ in all of them.
On the other hand, 7 - 12 have the same \mrorder for similar reason;
${\tt cde}$ is prior to ${\tt abc}$.

The size of the grammar generated by \rp varies according to
how to select a pair when there are several distinct most frequent pairs that overlap.
For instance,
consider the text
{\tt bcxdabcyabzdabvbcuda}.
There are 3 most frequent pairs,
${\tt ab}$, ${\tt bc}$, and ${\tt da}$
with 3 occurrences.
If \rp takes ${\tt ab}$ first,
the rule set of generated grammar may become
$\{v_1 \rightarrow {\tt ab},
~v_2 \rightarrow {\tt bc},
~v_3 \rightarrow {\tt d}v_1,
~S \rightarrow 
v_2
{\tt x}
v_3
{\tt c}
{\tt y}
v_1
{\tt z}
v_3
{\tt v}
v_2
{\tt u}
{\tt d}
{\tt a}\}$
and the size of it is 19.
On the other hand,
if \rp takes ${\tt da}$ first,
the rule set of generated grammar may become
$\{v_1 \rightarrow {\tt da},
~v_2 \rightarrow {\tt bc},
~S \rightarrow 
v_2
{\tt x}
v_1
v_2
{\tt y}
{\tt a}
{\tt b}
{\tt z}
v_1
{\tt b}
{\tt v}
v_2
{\tt u}
v_1
\}$
and the size of it is 18.

\begin{Remark}\label{rem:order}
	If there are several distinct pairs with the same maximum frequency,
	the size of the grammar generated by \rp depends on the replacement order of them.
\end{Remark}

However, the following theorem states that MR-order rather than the replacement order of pairs is essentially important for the size of the grammar generated by \rp.

\begin{theorem}\label{theo:order}
	The sizes of the grammars generated by \rp are the same
	if they are generated in the same \mrorder.
\end{theorem} 

\begin{proof}
	Let $T$ be a variable sequence appearing in the grammar generation process of \rp, and
	$f$ be the maximum frequency of pairs in $T$.
	Suppose that $T'$ is a variable sequence generated after \rp replaces every pair occurring $f$ times.
	By Theorem~\ref{theo:rpmr},
	all generated $T'$ are isomorphic to one another, then the length of all of them is the same,
	regardless of the replacement order of pairs.
	Let $r_1$ be the most frequent maximal repeats of $T$ such that $r_1$ is prior to all other ones in this \mrorder.
	$r_1$ is converted into a variable as a result, and
	by Lemma~\ref{lem:pm}, all pairs included in $r_1$ are distinct.
	Then, the size of the subgrammar which exactly derives $r_1$ is $2(|r_1|-1)+1=2|r_1|-1$.
	This holds for the next prioritized maximal repeat (we denote it by $r_2$) with a little difference;
	the pattern actually replaced would be a substring of $r_2$
	excluding the beginning or the end of it, if there are occurrences of overlap with $r_1$.
	However, these strings are common in the same \mrorder,
	then the sizes of generated subgrammars are the same,
	regardless of the selecting order of pairs.
	This similarly holds for all of the most frequent maximal repeats,
	for every maximum frequency of pairs,
	through the whole process of \rp.
\end{proof} 

\subsection{Greatest size difference of \rp}

We consider the problem of how large the difference between possible outcomes of \rp can be.

\begin{Definition}[Greatest Size Difference]
	Let $g$ and $g'$ be sizes of any two possible grammars that can be generated by \rp for a given text.
	Then, the Greatest Size Difference of \rp (GSDRP) is ${\rm max}(|g - g'|)$.
\end{Definition} 

We show a lower bound of above GSDRP in the following theorem.

\begin{theorem}\label{theo:gsgrpomega}
	Given a text with length $n$,
	a lower bound of GSDRP is $\frac{1}{6}(\sqrt{6n+1}+13)$.
\end{theorem} 

\begin{proof}
	Let $B$, $L$, and $R$ be strings such that
	\begin{align*}
		B &= l_1 xy r_1 l_2 xy r_2 
		\cdots \mathit{l_{f{\rm -1}}} xy \mathit{r_{f{\rm -1}}} \mathit{l_f} xy \mathit{r_f}, \\
		L &= \diamondsuit l_1 x\diamondsuit l_2 x \cdots \diamondsuit \mathit{l_f} x, ~~~\\
		R &= \diamondsuit y r_1\diamondsuit y r_2 \cdots \diamondsuit y \mathit{r_f},
	\end{align*} 
  where
  $x,y,l_1, \ldots, l_f, r_1, \ldots, r_f$ denote distinct symbols,
  and each occurrence of $\diamondsuit$ denotes a distinct symbol.
  Consider text $T = BL^{f-1}R^{f-1}$.
	Here, 
	$xy$, $l_1 x$, $\cdots$, $l_f x$, $\mathit{y r_1}$, $\cdots$, $\mathit{y r_f}$
  are the most frequent maximal repeats with frequency $f$ in $T$.
	Let $G$ and $G'$ be grammars generated by \rp for $T$ in different \mrorder, such that
	(i) $xy$ is prior to all other maximal repeats, and
	(ii) $xy$ is posterior to all other maximal repeats, respectively.
	We denote the sizes of $G$ and $G'$ by $g$ and $g'$, respectively.

	First, we consider $G$ and how \rp generates it.
	The first rule generated by replacement is $v_1\rightarrow xy$ because of the \mrorder.
	After replacement, $L$ and $R$ is unchanged but $B$ becomes the following text:
	\begin{equation*}
		B_1 = l_1 v_1 r_1 l_2 v_1 r_2 
		\cdots \mathit{l_{f{\rm -1}}} v_1 \mathit{r_{f{\rm -1}}} \mathit{l_f} v_1 \mathit{r_f}.
	\end{equation*} 
    Each pair in $B_1$ occurs only once in the whole text $B_1L^{f-1}R^{f-1}$.
	This means that $B_1$ is never be shortened from the current length, $3f$.
	In the remaining steps, each $\mathit{l_i}x$ and $y\mathit{r_i}$ (for $i = 1, \cdots, f$) is replaced,
	and that is all.
	$L$ and $R$ changed to texts whose length are both $2f$.
	Hence, the following holds:
	\begin{equation}
		g = 3f+2\cdot 2f+2(1+2f)=11f+2.\label{eq:size1}
	\end{equation} 

	Next, we consider $G'$ and how \rp generates it.
	By its \mrorder, 
	$l_1 x$, $\cdots$, $l_f x$, $\mathit{y r_1}$, $\cdots$, $\mathit{y r_f}$
	are replaced before $xy$ is selected.
	They do not overlap with each other, and after they are replaced,
	$xy$ does not occur in the generated text.
	Therefore, in $G'$,
	there are $2f$ rules which derive 
	each $\mathit{l_i}x$ and $y\mathit{r_i}$ (for $i = 1, \cdots, f$), respectively,
	but the rule which derives $xy$ is absent.
	$L$ and $R$ changed to texts whose length are both $2f$,
	and $B$ changed to a text with length $2f$.
	Hence, the following holds:
	\begin{equation}
		g' = 2f+2\cdot 2f+2\cdot2f=10f.\label{eq:size2}
	\end{equation} 

	Let us denote the length of the original text $T=BL^{f-1}R^{f-1}$ by $n$.
	Then, the following holds:
	\begin{align*}
		n = 4f + 2(3f)(f-1) = 6f^2-2f.~~~
	\end{align*} 
	Therefore, 
	\begin{align}
		f = \frac{1}{6}(\sqrt{6n+1}+1)\label{eq:length}
	\end{align} 
	holds.
	By~(\ref{eq:size1}),~(\ref{eq:size2}),~and~(\ref{eq:length}),
	\begin{align*}
		g - g' &= 11f+2 - 10f = f+2\\
		&=\frac{1}{6}(\sqrt{6n+1}+13)
	\end{align*} 
	holds and the theorem follows.
\end{proof}

\section{\mr}\label{sec:mrrepair}

The main strategy of our proposed method is to
recursively replace the most frequent maximal repeats,
instead of the most frequent pairs.

In this section,
first,
we explain the na\"{\i}ve version of our method called \nmr.
While it can have bad performance in specific cases,
it is simple and helpful to understand our main result.
Then, we describe our proposed \mr.

\subsection{\nmr}

\begin{Definition}[\nmr]\label{def:nmr}
  For input text $T$,
	let $G=\set{V,\Sigma ,S,R}$ be the grammar generated by \nmr.
	\nmr constructs $G$ by the following steps:\\

	\noindent
	{\bf Step~1.} Replace each symbol $a\in \Sigma$ with a new variable $v_a$ and add $v_a\rightarrow a$ to $R$.\\
	{\bf Step~2.} Find the most frequent maximal repeat $r$ in $T$.\\
  {\bf Step~3.} Replace every occurrence (or, as many occurrences as possible, when there are overlaps) 
  of $r$ in $T$ with a new variable $v$,
	and then add $v \rightarrow r$ to $R$.\\
	{\bf Step~4.} Re-evaluate the frequencies of maximal repeats
	for the renewed text generated in {\bf Step~3}.
	If the maximum frequency is 1,
	add $S \rightarrow {\rm (current~text)}$ to $R$,
	and terminate.
	Otherwise, return to {\bf Step~2}. 
\end{Definition}

We can easily extend the concept of \mrorder
to this \nmr.

\begin{figure*}[tb]
	\centering
	\begin{tabular}{c|c|c|c|c|c|c|c|c|c|c|c|}
		\cline{2-12}
		                                                                                           & {\tt a}                    & {\tt b}           & {\tt r}           & {\tt a}           & {\tt c}           & {\tt a}           & {\tt d}           & {\tt a}           & {\tt b}           & {\tt r}           & {\tt a}           \\ \cline{2-12} \noalign{\vspace{2pt}} \hline
		\multicolumn{1}{|c||}{$v_{\alpha} \rightarrow \alpha~(\alpha = {\rm {\tt a,b,r,c,d}})$}
		                                                                                           & $v_{\rm {\tt a}}$          & $v_{\rm {\tt b}}$ & $v_{\rm {\tt r}}$ & $v_{\rm {\tt a}}$ & $v_{\rm {\tt c}}$ & $v_{\rm {\tt a}}$ & $v_{\rm {\tt d}}$ & $v_{\rm {\tt a}}$ & $v_{\rm {\tt b}}$ & $v_{\rm {\tt r}}$ & $v_{\rm {\tt a}}$ \\ \hline
		\multicolumn{1}{|c||}{$v_1 \rightarrow v_{\rm {\tt a}}v_{\rm {\tt b}}v_{\rm {\tt r}}v_{\rm {\tt a}}$}
		                                                                                           & \multicolumn{4}{c|}{$v_1$} & $v_{\rm {\tt c}}$ & $v_{\rm {\tt a}}$ & $v_{\rm {\tt d}}$
		                                                                                           & \multicolumn{4}{c|}{$v_1$}                                                                                                                                                                                                         \\ \hline
		\multicolumn{1}{|c||}{$S \rightarrow v_1v_{\rm {\tt c}}v_{\rm {\tt a}}v_{\rm {\tt d}}v_1$} & \multicolumn{11}{c|}{$S$}                                                                                                                                                                                                          \\ \hline
	\end{tabular}
	\caption{An example of the grammar generation process of \nmr for text {\tt abracadabra}.
	The generated grammar is
	$\set{
		\set{v_{\rm {\tt a}},v_{\rm {\tt b}},v_{\rm {\tt r}},v_{\rm {\tt c}},v_{\rm {\tt d}}, v_1, S},
		\set{{\rm {\tt a},{\tt b},{\tt r},{\tt c},{\tt d}}}, S,
		\set{v_{\rm {\tt a}}\rightarrow {\rm {\tt a}},
		v_{\rm {\tt b}}\rightarrow {\rm {\tt b}},
		v_{\rm {\tt r}}\rightarrow {\rm {\tt r}},
		v_{\rm {\tt c}}\rightarrow {\rm {\tt c}},
		v_{\rm {\tt d}}\rightarrow {\rm {\tt d}},
		v_1 \rightarrow v_{\rm {\tt a}}v_{\rm {\tt b}}v_{\rm {\tt r}}v_{\rm {\tt a}},
		S \rightarrow v_1v_{\rm {\tt c}}v_{\rm {\tt a}}v_{\rm {\tt d}}v_1}
		}$,
	and the grammar size is $14$.
	}
	\label{fig:abranmr}
\end{figure*}

We show an example of the grammar generation
process of \nmr in Figure~\ref{fig:abranmr}.
Figure~\ref{fig:abrarp}~and~\ref{fig:abranmr} illustrates
the intuitive reason why the strategy using maximal repeats is effective
compared to that using pairs.
When compressing text
$v_{\tt a}v_{\tt b}v_{\tt r}v_{\tt a}v_{\tt c}v_{\tt a}v_{\tt d}v_{\tt a}v_{\tt b}v_{\tt r}v_{\tt a}$,
\rp and \nmr both generate subgrammars which derive the most frequent maximal repeat
$v_{\tt a}v_{\tt b}v_{\tt r}v_{\tt a}$.
The rule set of the subgrammar of \rp is
$\{v_1 \rightarrow v_{\tt a}v_{\tt b}, v_2 \rightarrow v_1v_{\tt r}, v_3 \rightarrow v_2v_{\tt a} \}$,
and the size is 6.
On the other hand,
the rule set of subgrammar of \nmr is
$\{v_1 \rightarrow v_{\tt a}v_{\tt b}v_{\tt r}v_{\tt a}\}$,
and the size is 4.

However, the following theorem indicates that
the size of the grammar generated by \nmr is larger than that by \rp
in particular cases, even when they work in the same \mrorder.

\begin{theorem}\label{theo:naive}
	Given a text $T$ with length $n$, and
	assume that \rp and \nmr work in the same \mrorder.
	Let $\gsize{rp}$ and $\gsize{nmr}$ be sizes of
	grammars generated by \rp and \nmr for $T$, respectively.
	Then, there is a case where $\gsize{nmr} = \gsize{rp} + \ord(\log n)$ holds.%
	\footnote{We show a concrete example of this theorem in Appendix.}
\end{theorem} 

\begin{proof}
	Assume that $\gram{rp} = \set{\var{rp}, \sym{rp}, \stt{rp}, \rul{rp}}$
	and $\gram{nmr} = \set{\var{nmr}, \sym{nmr}, \stt{nmr}, \rul{nmr}}$
	are grammars generated by \rp and \nmr, respectively.
	Let $T' = v_1\cdots v_n$ such that
	$v_i\in\var{rp}\cap\var{nmr}$ and $v_i \rightarrow T[i]\in\rul{rp}\cap\rul{nmr}$ (for $i = 1,\cdots , n$),
	and
	$\sgram{rp} = \set{\svar{rp}, \ssym{rp}, \sstt{rp}, \srul{rp}}$
	(or $\sgram{nmr} = \set{\svar{nmr}, \ssym{nmr}, \sstt{nmr}, \srul{nmr}}$)
	be a subgrammar of $\gram{rp}$ (or $\gram{nmr}$)
	which derives $T'$.
	Assume that
	$T' = (uw)^{2^{m+1}-1} u$,
	where $u\in \var{rp}\cap\var{nmr}$, $w\in (\var{rp}\cap\var{nmr})^{+}$
	such that $uwu$ is the most frequent maximal repeat of $T'$,
	and $m \in \mathbb{N}^{+}$.
	Note that $2^{m+1}-1 = \sum_{i=0}^{m}2^i$.
	Here $\srul{rp}$ and $\srul{nmr}$ consist as follows:
	\vspace{6pt}
	\begin{description}
		\item[$\srul{rp}$:]
		      Assume that $x_i\in \svar{rp}$ for $1\le i\le m$
		      and $y_j\in \svar{rp}\cup\ssym{rp}$ for $1\le j\le |w|$, then
		      \begin{itemize}
			      \item $|w|$ rules
			            $y_j\rightarrow y_ly_r$ with $\derive{y_{|w|}} = uw$.
			      \item
			            One rule $x_1\rightarrow y_{|w|}y_{|w|}$ and
			            $\log_2{\lfloor 2^{m+1}-1\rfloor}-1 = m-1$ rules
			            $x_i\rightarrow x_{i-1}x_{i-1}$ for $2\le i\le m$.
			      \item One rule
			            $\sstt{rp}\rightarrow x_{m} x_{m-1}\cdots x_{1}y_{|w|}$.
		      \end{itemize}
		\item[$\srul{nmr}$:]
		      Assume that $d = |\svar{nmr}| = |\srul{nmr}|$ and $z_i\in \svar{nmr}$ for $1\le i\le d$, then
		      \begin{itemize}
			      \item One rule $z_1 \rightarrow uwu$.
			      \item $d - 1$ rules
			            $z_i \rightarrow z_{i-1}wz_{i-1}$ for $2\le i\le d$
			            and $z_d = \sstt{nmr}$.
		      \end{itemize}
	\end{description}
	\vspace{6pt}
	Let $\sgsize{rp}$ and $\sgsize{nmr}$ be sizes of $\sgram{rp}$ and $\sgram{nmr}$, respectively.
	Then, the following holds:
	\begin{align}
		\sgsize{rp}  & = 2|w|+2m+(m+2)=3m+2|w|+2\label{eq:rpsize}        \\
		\sgsize{nmr} & = |w|+2+(|w|+2)(d-1) = (|w|+2)d\label{eq:nmrsize}
	\end{align}
	Here, with regard to the length of $T'$, the following holds:
	\begin{equation*}
		(2^d-1)|w|+2^d = n = (2(2^m-1)+1)(|w|+1)+1.
	\end{equation*}
	Since the right-side results in $2^{m+1}|w|+2^{m+1}$,
	$d = m+1$ follows it.
	Hence, by equation~(\ref{eq:rpsize}) and (\ref{eq:nmrsize}), the following holds:
	\begin{equation*}
		\sgsize{nmr}-\sgsize{rp}=(m - 1)(|w| - 1) - 1.
	\end{equation*}
	Therefore, $\sgsize{nmr} > \sgsize{rp}$ holds for some $(m, |w|)$,
	and the proposition holds.
\end{proof}

\subsection{\mr}

The reason why the grammar size of \nmr becomes larger than that of \rp as seen in Theorem~\ref{theo:naive} is
that \nmr cannot replace all occurrences of the most frequent maximal repeats
if it overlaps with another occurrence of itself.
In the remainder of this section,
we describe \mr, which is an improved version of the above \nmr.

\begin{Definition}[\mr]\label{def:mr}
  For input text $T$,
  let $G=\set{V,\Sigma ,S,R}$ be the grammar generated by \mr.
  \mr constructs $T$ by the following steps:\\

\vspace{-10pt}
	\noindent
	{\bf Step~1.} Replace each symbol $a\in \Sigma$ with a new variable $v_a$ and add $v_a\rightarrow a$ to $R$.\\
	{\bf Step~2.} Find the most frequent maximal repeat $r$ in $T$.\\
	{\bf Step~3.} Check if $|r|>2$ and $r[1]=r[|r|]$, and
	if so, replace $r$ with $r[2..|r|]$. \\
	{\bf Step~4.} Replace every occurrence of $r$ with a new variable $v$,
	then add $v \rightarrow r$ to $R$.\\
	{\bf Step~5.} Re-evaluate the frequencies of maximal repeats
	for the renewed text generated in {\bf Step~4}.
	If the maximum frequency is 1,
	add $S \rightarrow {\rm (current~text)}$ to $R$, and terminate.
	Otherwise, return to {\bf Step~2}. 
\end{Definition}

\begin{figure*}[tb]
	\centering
	\begin{tabular}{c|c|c|c|c|c|c|c|c|c|c|c|}
		\cline{2-12}
		                                                                                           & {\tt a}                    & {\tt b}           & {\tt r}           & {\tt a}           & {\tt c}           & {\tt a}           & {\tt d}           & {\tt a}           & {\tt b}           & {\tt r}           & {\tt a}           \\ \cline{2-12} \noalign{\vspace{2pt}} \hline
		\multicolumn{1}{|c||}{$v_{\alpha} \rightarrow \alpha~(\alpha = {\rm {\tt a,b,r,c,d}})$}
		                                                                                           & $v_{\rm {\tt a}}$          & $v_{\rm {\tt b}}$ & $v_{\rm {\tt r}}$ & $v_{\rm {\tt a}}$ & $v_{\rm {\tt c}}$ & $v_{\rm {\tt a}}$ & $v_{\rm {\tt d}}$ & $v_{\rm {\tt a}}$ & $v_{\rm {\tt b}}$ & $v_{\rm {\tt r}}$ & $v_{\rm {\tt a}}$ \\ \hline
		\multicolumn{1}{|c||}{$v_1 \rightarrow v_{\rm {\tt a}}v_{\rm {\tt b}}v_{\rm {\tt r}}$}
		                                                                                           & \multicolumn{3}{c|}{$v_1$} & $v_{\rm {\tt a}}$ & $v_{\rm {\tt c}}$ & $v_{\rm {\tt a}}$ & $v_{\rm {\tt d}}$
		                                                                                           & \multicolumn{3}{c|}{$v_1$} & $v_{\rm {\tt a}}$                                                                                                                                                                                     \\ \hline
		\multicolumn{1}{|c||}{$v_2 \rightarrow v_1v_{\rm {\tt a}}$}
		                                                                                           & \multicolumn{4}{c|}{$v_2$} & $v_{\rm {\tt c}}$ & $v_{\rm {\tt a}}$ & $v_{\rm {\tt d}}$
		                                                                                           & \multicolumn{4}{c|}{$v_2$}                                                                                                                                                                                                         \\ \hline
		\multicolumn{1}{|c||}{$S \rightarrow v_2v_{\rm {\tt c}}v_{\rm {\tt a}}v_{\rm {\tt d}}v_2$} & \multicolumn{11}{c|}{$S$}                                                                                                                                                                                                          \\ \hline
	\end{tabular}
	\caption{An example of the grammar generation process of \mr for text {\tt abracadabra}.
	The generated grammar is
	$\set{
		\set{v_{\rm {\tt a}},v_{\rm {\tt b}},v_{\rm {\tt r}},v_{\rm {\tt c}},v_{\rm {\tt d}}, v_1, S},
		\set{{\rm {\tt a},{\tt b},{\tt r},{\tt c},{\tt d}}}, S,
		\set{v_{\rm {\tt a}}\rightarrow {\rm {\tt a}},
		v_{\rm {\tt b}}\rightarrow {\rm {\tt b}},
		v_{\rm {\tt r}}\rightarrow {\rm {\tt r}},
		v_{\rm {\tt c}}\rightarrow {\rm {\tt c}},
		v_{\rm {\tt d}}\rightarrow {\rm {\tt d}},
		v_1 \rightarrow v_{\rm {\tt a}}v_{\rm {\tt b}}v_{\rm {\tt r}},
		v_2 \rightarrow v_1v_{\rm {\tt a}},
		S \rightarrow v_2v_{\rm {\tt c}}v_{\rm {\tt a}}v_{\rm {\tt d}}v_2}
		}$,
	and the grammar size is $15$.
	}%
	\vspace{-10pt}
	\label{fig:abramr}
\end{figure*}

We show an example of the grammar generation process of \mr in Figure~\ref{fig:abramr}.
We can easily extend the concept of \mrorder to this \mr.
We do not care if it uses $r[1..|r-1|]$ in {\bf Step 3}, instead of $r[2..|r|]$.
\mr can replace all occurrences of $r$
even if it overlaps with itself in some occurrences,
since by Lemma~\ref{lem:lenol},
the length of overlaps of the most frequent maximal repeats is at most 1.
If $r[1] = r[|r|]$ but $r$ does not overlap with itself,
then $r[1]v$ becomes the most frequent maximal repeat after $r[2..|r|]$ is replaced by $v$,
and $r[1]v$ would be replaced immediately.
\mr still cannot replace all of them if $|r|=2$,
but the same is said to \rp.

We show an example of the grammar generation process of \mr in Figure~\ref{fig:abramr}.
Although the size of the generated grammar in Figure~\ref{fig:abramr} is larger than that of \nmr
shown in Figure~\ref{fig:abranmr},
it is still smaller than that of \rp shown in Figure~\ref{fig:abrarp}.

\begin{theorem}\label{theo:ratio}
	Assume that \rp and \mr work based on the same \mrorder for a given text.
	Let $\gsize{rp}$ and $\gsize{mr}$ be sizes of grammars
	generated by \rp and \mr, respectively.
	Then,
	$\frac{1}{2} \gsize{rp} < \gsize{mr} \le \gsize{rp}$ holds.
\end{theorem}

\begin{proof}
	Assume that
	$\gram{rp} = \set{\var{rp}, \sym{rp}, \stt{rp}, \rul{rp}}$
	and $\gram{mr} = \set{\var{mr}, \sym{mr}, \stt{mr}, \rul{mr}}$
	are grammars generated by \rp and \mrn, respectively,
	for a given text $T$ with length $n$.
	Let $T' = v_1\cdots v_n$ such that
	$v_i\in\var{rp}\cap\var{mr}$ and $v_i \rightarrow T[i]\in\rul{rp}\cap\rul{mr}$ (for $i = 1,\cdots , n$).

	We start with $T'$.
	Let $f_1$ be the maximum frequency of maximal repeats in $T'$.
	By Corollary~\ref{coro:pm}, the maximum frequency of pairs in $T'$ is also $f_1$.
	Let $\sgram{rp}^{(f_1)}$ (or $\sgram{mr}^{(f_1)}$) be a subgrammar of $\gram{rp}$ (or $\gram{mr}$)
	which is generated while \rp (or \mr) replaces pairs (or maximal repeats) with frequency $f_1$,
	$\sgsize{rp}^{(f_1)}$ (or $\sgsize{mr}^{(f_1)}$) be the size of it, and
	$T^{(f_1)}_{\mathit{rp}}$ (or $T^{(f_1)}_{\mathit{mr}}$) be the renewed text
	after all pairs (or maximal repeats) with frequency $f_1$ are replaced.
	Let $r_1^{(f_1)}, \cdots, r_{m_1}^{(f_1)}$ be maximal repeats with frequency $f_1$ in $T'$,
	and assume that they are prioritized in this order by the \mrorder.
	Let each $l_i^{(f_1)}$ (for $i = 1,\cdots ,m_1$) be the length of the longest substring of $r_i^{(f_1)}$
	such that there are variables that derive the substring in both $\sgram{rp}^{(f_1)}$ and $\sgram{mr}^{(f_1)}$.
	Note that this substring is common to \rp and \mr,
	and each $l_i^{(f_1)}$ is at least 2.
	Then, by Lemma~\ref{lem:pm}, the following holds:
	\begin{align*}
		\sgsize{rp}^{(f_1)} &= \sum_{i=1}^{m_1} 2(l_i^{(f_1)}-1)~,~~~\\
		\sgsize{mr}^{(f_1)} &= \sum_{i=1}^{m_1} l_i^{(f_1)}. ~~~
	\end{align*}
	Therefore,
	\begin{align}
		\therefore ~ \frac{1}{2}\sgsize{rp}^{(f_1)} < \sgsize{mr}^{(f_1)} \le \sgsize{rp}^{(f_1)}\label{eq:f1}
	\end{align}
	holds.
	The renewed texts $T^{(f_1)}_{\mathit{rp}}$ and $T^{(f_1)}_{\mathit{mr}}$ are isomorphic for $\var{rp}$ and $\var{mr}$.
	Let $f_2$ be the maximum frequency of maximal repeats in $T^{(f_1)}_{\mathit{rp}}$
	(and this is the same in $T^{(f_1)}_{\mathit{mr}}$).
	Then, the similar discussion holds for $\sgram{rp}^{(f_2)}$ and $\sgram{mr}^{(f_2)}$.
	Hence, $\frac{1}{2}\sgsize{rp}^{(f_2)} < \sgsize{mr}^{(f_2)} \le \sgsize{rp}^{(f_2)}$ holds
	similarly to (\ref{eq:f1}),
	and the renewed texts $T^{(f_2)}_{\mathit{rp}}$ and $T^{(f_2)}_{\mathit{mr}}$ are isomorphic.
	Inductively, for every maximum frequency of maximal repeats $f_i$,
	$\frac{1}{2}\sgsize{rp}^{(f_i)} < \sgsize{mr}^{(f_i)} \le \sgsize{rp}^{(f_i)}$ holds
	and the renewed texts $T^{(f_i)}_{\mathit{rp}}$ and $T^{(f_i)}_{\mathit{mr}}$ are isomorphic.
	Let $k$ be a natural number such that $f_k > 1$ and $f_{k+1} = 1$,
	which is the number of decreasing of maximum frequency through the whole process of \rp and \mr.
	Then,
	\begin{align}
		\gsize{rp} & = \sum_{j=1}^{k}\sgsize{rp}^{(f_j)} + |\Sigma| + |T^{(f_k)}_{\mathit{rp}}| 
		           = \sum_{j=1}^{k}\sum_{i=1}^{m_j} 2(l_i^{(f_j)}-1) + |\Sigma| + |T^{(f_k)}_{\mathit{rp}}|~,\label{eq:grp} \\
		\gsize{mr} & = \sum_{j=1}^{k}\sgsize{mr}^{(f_j)} + |\Sigma| + |T^{(f_k)}_{\mathit{mr}}| 
				   = \sum_{j=1}^{k}\sum_{i=1}^{m_j} l_i^{(f_j)} + |\Sigma| + |T^{(f_k)}_{\mathit{mr}}|~\label{eq:gmr}
	\end{align}
	holds.
	Because every $l_i^{(f_j)} \ge 2$ and $|T^{(f_k)}_{\mathit{rp}}| = |T^{(f_k)}_{\mathit{mr}}|$,
	$\frac{1}{2}\gsize{rp} < \gsize{mr} \le \gsize{rp}$
	follows (\ref{eq:grp}) and (\ref{eq:gmr})
	and the proposition holds.
	$\gsize{mr} = \gsize{rp}$ holds when every length of $l_i^{(f_j)}$ is 2.
\end{proof}

The following theorem shows that
unless the \mrorder of \rp and \mr are the same,
the grammar generated by \mr might be larger than that by \rp.

\begin{theorem}
	Unless the \mrorder of \rp and \mr are the same,
	there is a case where
	the size of the generated grammar by \mr becomes larger than that by \rp.
\end{theorem}

\begin{proof}
	As shown in Theorem~\ref{theo:ratio},
	the size of grammar generated by \mr would be strictly equal to that by \rp with the same \mrorder.
    Thus, we can reduce this problem to
    the problem that there is a difference between sizes of possible grammars generated by \rp as stated in Remark~\ref{rem:order}.
	Hence, there are the cases stated in the proposition
	if the \mrorder of \mr matches with a \mrorder of \rp which does not generate the smallest \rp grammar.
\end{proof} 

We can implement \mr by extending the original implementation of \rp
stated in~\cite{892708},
holding the same complexity.

\begin{theorem}\label{theo:ts}
	Let $G=\set{V,\Sigma ,S,R}$ be the generated grammar by \mr for a given text with length $n$.
	Then,
	\mr works in $\ord (n)$ expected time and
	$5n + 4k^2 + 4k^{\prime} + \lceil \sqrt{n + 1} \rceil - 1$ word space,
	where $k$ and $k'$ are the cardinalities of $\Sigma$ and $V$, respectively.
\end{theorem}

\begin{proof}
	Compared with \rp,
	the additional operations which \mr does in our implementation are
	(i) it extends the selected pair to left and right until it becomes a maximal repeat, and
	(ii) it checks and excludes either of the beginning or the end of the obtained maximal repeat if they are the same.
	They can be realized by only using the same data structures as that of \rp.
	Then, the space complexity of \mr follows Lemma~\ref{lem:rp}.

	We can clearly execute operation (ii) in constant time.
	So we consider how the time complexity is affected by operation (i).
	Let $l$ be the length of the maximal repeat containing the focused pair,
	and $f$ be the frequency of the pair.
	Then,
	when \mr checks the left- and right-extensions for all occurrences of the focused pair,
	$\ord(fl)$ excessive time is required compared with \rp.
	However, the length of entire text is shortened at least $f(l-1)$ by the replacement.
	Therefore, according to possible counts of replacement through the entire steps of the algorithm,
	\mr works in $\ord(n)$ expected time.
\end{proof}

\begin{Remark}
	We can convert a grammar of \rp to that of \mr by repeating the following transform: 
	If a variable $v$ appears only once on the right-hand side of other rules, 
	remove the rule for $v$ and replace the one occurrence of $v$ with the right-hand side of the removed rule. 
	However, time and space complexity stated in Theorem~\ref{theo:ts} cannot be achieved in this manner,
	since additional operations and memory for searching and storing such variables are required.
\end{Remark}

\section{Experiments}\label{sec:exp}

We implemented \mr and measured the number of generated rules and the execution time
in order to compare it to existing \rp implementations and
Re-PairImp\footnote{\url{https://bitbucket.org/IguanaBen/repairimproved}}
proposed by Ga{\'n}czorz and Je{\.z}~\cite{ganczorz2017improvements}.

As stated in Remark~\ref{rem:order},
the size of a generated grammar depends on the MR-order.
In practice, the MR-order varies how we implement the priority queue managing symbol pairs.
To see this, we used five \rp implementations in the comparison;
they were implemented by
Maruyama\footnote{\url{https://code.google.com/archive/p/re-pair/}},
Navarro\footnote{\url{https://www.dcc.uchile.cl/~gnavarro/software/index.html}},
Prezza\footnote{\url{https://github.com/nicolaprezza/Re-Pair}}~\cite{7921912},
Wan\footnote{\url{https://github.com/rwanwork/Re-Pair}; We ran it with level 0 (no heuristic option).},
and Yoshida%
\footnote{\url{https://github.com/syoshid/Re-Pair-VF}; We removed a routine to find the best rule set.}.

Table~\ref{tbl:used_text} summarizes the details of the texts we used in the comparison.
We used three texts as highly repetitive texts;
one is a randomly generated text (rand77.txt),
and the others are a Fibonacci string (fib41)
and a German text (einstein.de.txt) which were selected from Repetitive Corpus of Pizza\&Chili Corpus%
\footnote{\url{http://pizzachili.dcc.uchile.cl/repcorpus.html}}.
The randomly generated text, rand77.txt, consists of alphanumeric symbols and some special symbols;
and it is generated by concatenating 32 copies of a block that includes 1024 random patterns of length 64,
i.e., the size is $64\times 1024\times 32=2,097,152$ byte.
In addition, we used three texts (E.coli, bible.txt, world192.txt) for real data selected from Large Corpus%
\footnote{\url{http://corpus.canterbury.ac.nz/descriptions/\#large}}.
We executed each program seven times for each text and measured the elapsed CPU time only for grammar generation process.
We calculated the average time of the five results excluding the minimum and maximum values among seven.
We ran our experiments on a workstation equipped with an Intel(R) Xeon(R) E5-2670 2.30GHz dual CPU with 64GB RAM, running on Ubuntu 16.04LTS on Windows 10.
All the programs are compiled by gcc version 7.3.0 with ``-O3'' option.

\begin{table*}[tb]
	\begin{center}
		\caption{Text files used in our experiments.}%
		\vspace{12pt}
		\label{tbl:used_text}
		\begin{tabular}{l@{\quad}r@{\quad}r@{\quad}l}
			\hline
			\multicolumn{1}{c}{texts~~~~} & \multicolumn{1}{c}{size (bytes)~~} & \multicolumn{1}{c}{$|\Sigma|$~~} & \multicolumn{1}{c}{contents~~~~}                      \\
			\hline
			rand77.txt                    & 2,097,152                         & 77                               & 32 copies of 1024 random patterns of length 64        \\
			fib41                         & 267,914,296                       & 2                                & Fibonacci string from Pizza\&Chili Corpus             \\
			einstein.de.txt               & 92,758,441                        & 117                              & Edit history of Wikipedia for Albert Einstein \\
			E.coli                        & 4,638,690                         & 4                                & Complete genome of the E. Coli bacterium              \\
			bible.txt                     & 4,047,392                         & 63                               & The King James version of the bible                   \\
			world192.txt                  & 2,473,400                         & 94                               & The CIA world fact book                               \\
			\hline
		\end{tabular}
	\end{center}
	\vspace{-12pt}
\end{table*}

Table~\ref{tbl:exp_res_v2} lists the experimental results.
Here, we excluded the number of rules that generate a single terminal symbol from the number of rules
because it is the same between \mr and \rp.
As shown in the table,
for all texts except for fib41,
the size of rules generated by each \rp implementation differs from each other.%
\footnote{We found that the results of Yoshida were the same as those of Maruyama because Yoshida utilized the code of Maruyama.}.
In any case, \mr is not inferior to \rp in the size of rules.
For rand77.txt in particular,
the number of rules decreased to about 11\% and the size of rules decreased to about 55\%.
For einstein.de.txt, moreover,
the number of rules decreased to about 44\% and the size of rules decreased to about 72\%.
On the other hand, for the texts of Large Corpus, which are not highly repetitive,
it turned out that the effect of improvement was limited.
Note that fib41 does not contain any maximal repeats longer than 2 without overlaps.
Therefore, \mr generates the same rules as \rp.
Also note that \mr runs at a speed comparable to the fastest implementation of \rp.

\begin{table*}[t]
\begin{center}
  \caption{
    The sizes of generated grammars and the execution times.
    Each cell in the table represents
    the number of generated rules,
    the total lengths of the right side of all the rules except for the start variable,
    the length of the right side of the start variable,
    and the total grammar size
    in order from the top row.
	The fifth row separated by a line represents the execution time with seconds.
  }
\vspace{10pt}
\label{tbl:exp_res_v2}
\begin{tabular}{l|rrrrrrr}
\hline
               & \multicolumn{5}{l}{RePair}   & \multicolumn{1}{l}{Re-PairImp} & \multicolumn{1}{l}{MR-RePair} \\
	text file      & \multicolumn{1}{l}{Maruyama} & \multicolumn{1}{l}{Navarro} & \multicolumn{1}{l}{Prezza} & \multicolumn{1}{l}{Wan} 
	& \multicolumn{1}{l}{Yoshida} & \multicolumn{1}{l}{}           & \multicolumn{1}{l}{}          \\
\hline
rand77.txt     & 41,651    & 41,642        & 41,632     & 41,675    & 41,651      & 41,661     & {\bf 4,492}        \\
               & 83,302    & 83,284        & 83,264     & 83,350    & 83,302      & 83,322     & {\bf 46,143}       \\
			   & 9         & {\bf 2}       & 7          & {\bf 2}   & 9           & {\bf 2}    & 9            \\
			   & 83,311    & 83,286        & 83,271     & 83,352    & 83,311      & 83,324     & {\bf 46,152} \\
\cline{2-8}
			   & 0.41      & {\bf 0.37}    & 4.76       & 4.27      & 0.40        & 3.95       & 0.42         \\
\hline
fib41          & 38        & 38            & 38         & 38        & 38          & {\bf 37}         & 38           \\
			   & 76        & 76            & 76         & 76        & 76          & {\bf 74}         & 76           \\
			   & {\bf 3}   & {\bf 3}       & {\bf 3}    & {\bf 3}   & {\bf 3}     & 23         &  {\bf 3}           \\
               & {\bf 79}  & {\bf 79}      & {\bf 79}   & {\bf 79}  & {\bf 79}    & 97         & {\bf 79}     \\
\cline{2-8}
			   & 26.75     & {\bf 23.94}   & 96.05      & 483.86    & 25.04       & 1360.40    & 33.62        \\
\hline
einstein.de.txt & 49,968    & 49,949        & 50,218     & 50,057    & 49,968      & 49,933     & {\bf 21,787}       \\
               & 99,936    & 99,898        & 100,436    & 100,114   & 99,936      & 99,866     & {\bf 71,709}       \\
               & 12,734    & 12,665        & 13,419     & {\bf 12,610}    & 12,734      & 12,672     & 12,683       \\
			   & 112,670   & 112,563       & 113,855    & 112,724   & 112,670     & 112,538    & {\bf 84,392} \\
\cline{2-8}
			   & 30.08     & 43.45         & 216.74     & 213.15    & 30.76       & 452.56     & {\bf 29.63}  \\
\hline
E.coli         & 66,664    & 66,757        & 66,660     & 67,368    & 66,664      & 66,739     & {\bf 62,363}       \\
               & 133,328   & 133,514       & 133,320    & 134,736   & 133,328     & 133,478    & {\bf 129,138}      \\
               & 651,875   & {\bf 649,660}       & 650,538    & 652,664   & 651,875     & 650,209    & 650,174      \\
			   & 785,203   & 783,174       & 783,858    & 787,400   & 785,203     & 783,687    & {\bf 779,312}\\
\cline{2-8}
			   & 1.20      & {\bf 1.02}    & 14.67      & 10.37     & 1.56        & 27.04      & 1.33         \\
\hline
bible.txt      & 81,193    & 81,169        & 80,999     & 81,229    & 81,193      & 81,282     & {\bf 72,082}       \\
               & 162,386   & 162,338       & 161,998    & 162,458   & 162,386     & 162,564    & {\bf 153,266}      \\
               & 386,514   & 386,381       & 386,992    & 386,094   & 386,514     & {\bf 385,989}    & 386,516      \\
			   & 548,900   & 548,719       & 548,990    & 548,552   & 548,900     & 548,553    & {\bf 539,782}\\
\cline{2-8}
			   & 1.33      & {\bf 1.21}    & 13.00      & 9.12      & 1.47        & 24.38      & 1.27         \\
\hline
world192.txt   & 55,552    & 55,798        & 55,409     & 55,473    & 55,552      & 55,437     & {\bf 48,601}       \\
               & 111,104   & 111,596       & 110,812    & 110,946   & 111,104     & 110,874    & {\bf 104,060}      \\
               & 213,131   & 213,962       & 213,245    & {\bf 212,647}   & 213,131     & 212,857    & 212,940      \\
			   & 324,235   & 325,558       & 324,057    & 323,593   & 324,235     & 323,731    & {\bf 317,000}\\
\cline{2-8}
			   & 0.59      & 0.80          & 7.57       & 4.89      & {\bf 0.56}  & 12.35      & 0.66         \\
\hline
\end{tabular}
\end{center}
\end{table*}

\section{Conclusion}\label{sec:concl}

	In this thesis,
    we analyzed \rp and showed
    that \rp replaces step by step the most frequent pairs within the corresponding most frequent maximal repeats.
    Motivated by this analysis,
	we designed a novel variant of \rp, called \mr,
	which is based on substituting the most frequent maximal repeats at once
    instead of substituting the most frequent pairs consecutively.
	Moreover, we implemented \mr and compared the grammar generated by it to that by \rp for several texts,
	and confirmed the effectiveness of \mr experimentally especially for highly repetitive texts.

    We defined the greatest size difference of any two possible grammars that can be generated by RePair for a given text, and named it as GSDRP.
    Moreover, we showed that a lower bound of GSDRP is $\frac{1}{6}(\sqrt{6n+1}+13)$ for a given text of length $n$.
    Improving the lower bound and showing a upper bound of GSDRP are our future works.



    Although we did not discuss how to encode grammars,
    it is a very important issue from a practical point of view.
    For \mr, if we simply use delimiters to store the rule set,
    the number of rules may drastically affect the compressed data size.
    To develop an efficient encoding method for \mr is one of our future works.

\section*{Acknowledgments}
The authors would like to thank the people who provided the source codes.
This work was supported by JSPS KAKENHI Grant Numbers JP17H06923, JP17H01697, JP16H02783, JP18H04098, and JP18K11149, Japan.
This work was also supported by JST CREST Grant Number JPMJCR1402, Japan.

%

\clearpage
\appendix
\section{Appendix}\label{appendix}

We show Figures~\ref{fig:apprp},~\ref{fig:appnmr},~and~\ref{fig:appmr} to help for understanding Proof of Theorem~\ref{theo:naive}.

Let $\gram{rp}$, $\gram{nmr}$, and $\gram{mr}$ be the grammars generated by \rpn, \nmrn, and \mrn, respectively.
For a given text $T = a_1\cdots a_n ~(a_i \in \Sigma, ~1\le i\le n)$ of length $|T| = n$,
let $\gsize{rp}$, $\gsize{nmr}$, and $\gsize{mr}$ be the sizes of $\gram{rp}$, $\gram{nmr}$, and $\gram{mr}$, respectively.
Here, assume that $T = (aw)^{2(2^m-1)+1} a$,
where $w\in \Sigma^{+}$ such that $awa$ is the most frequent maximal repeat of $T$,
and $m \in \mathbb{N}^{+}$.
Then, by Proof of Theorem~\ref{theo:naive}, $\gsize{nmr} > \gsize{rp}$ holds with
some $m$ and $w$ such that $(m - 1)(|w| - 1) > 1$.

We show a concrete example of the grammar generation process of \rp and $\gram{rp}$ 
for $T = ({\rm {\tt abcd}})^7{\rm {\tt a}}$ with $m = 2$ and $|w| = 3$ in Figure~\ref{fig:apprp}.
The size $\gsize{rp}$ is $18$ for this example.
We also show an example of the process of \nmr and $\gram{nmr}$ for the same $T$ in Figure~\ref{fig:appnmr}.
As we see, the size $\gsize{nmr}$ is $19$, and thus $\gsize{nmr} > \gsize{rp}$ holds.
As shown in Figure~\ref{fig:appnmr}, in particular cases,
\nmr may fail to extract repetitive patterns
(like {\tt abcd} of $({\rm {\tt abcd}})^7{\rm {\tt a}}$ for the running example).
However, this problem is solved by using \mrn.
We show an example of the process of \mr and $\gram{mr}$ for the same $T = ({\rm {\tt abcd}})^7{\rm {\tt a}}$ in Figure~\ref{fig:appmr}.
The size $\gsize{mr}$ is $16$ and this is smaller than $\gsize{rp} = 18$.
While the most frequent maximal repeat at the second replacement step is 
$v_{\rm {\tt a}}v_{\rm {\tt b}}v_{\rm {\tt c}}v_{\rm {\tt d}}v_{\rm {\tt a}}$,
\mr replaces
$v_{\rm {\tt a}}v_{\rm {\tt b}}v_{\rm {\tt c}}v_{\rm {\tt d}}$ with new variable $v_1$
because of the additional {\bf Step 3} of Definition~\ref{def:mr}.

\begin{landscape}
	\begin{figure}
	\includegraphics[angle=270, trim=0cm 0cm 17cm 0cm, clip, scale=0.85]{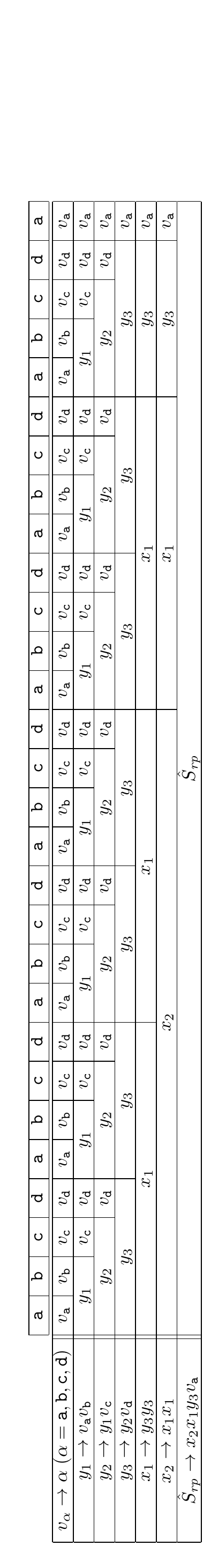}
		\caption{Grammar generation process of \rp and its generated grammar for text $({\rm {\tt abcd}})^{7}${\tt a}.
		The grammar size is $18$.}%
	\label{fig:apprp}
	\end{figure}

	\begin{figure}
	\includegraphics[angle=270, trim=0cm 0cm 18.2cm 0cm, clip, scale=0.85]{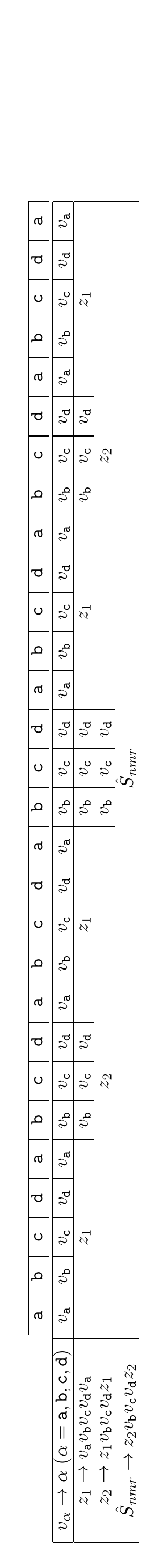}
		\caption{Grammar generation process of \nmr and its generated grammar for text $({\rm {\tt abcd}})^{7}${\tt a}.
		The grammar size is $19$.}%
	\label{fig:appnmr}
	\end{figure}

	\begin{figure}
	\includegraphics[angle=270, trim=0cm 0cm 17.8cm 0cm, clip, scale=0.85]{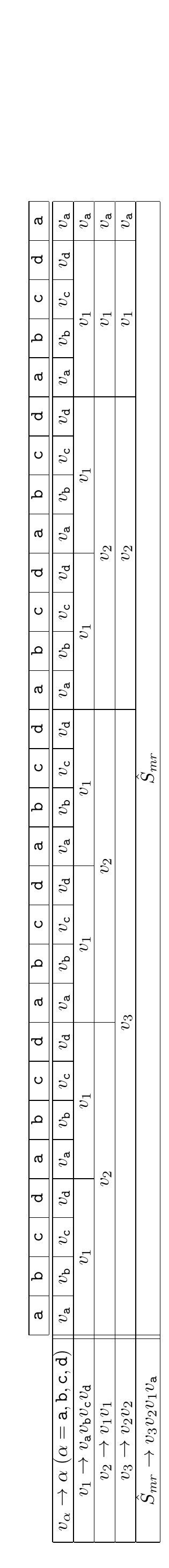}
		\caption{Grammar generation process of \mr and its generated grammar for text $({\rm {\tt abcd}})^{7}${\tt a}.
		The grammar size is $16$.}%
	\label{fig:appmr}
	\end{figure}
\end{landscape}

\end{document}